\documentclass[amssymb,12pt,longbibliography,nofootinbib]{revtex4-1}
\usepackage{graphicx}
\usepackage{appendix} 
\usepackage{hyperref}
\usepackage{bm,amsmath}
\usepackage{amsthm}
\usepackage{setspace}
\usepackage{slashed,mathrsfs}
\usepackage{xcolor}

\def\BLO{ B ({\cal H}, {\mathbb C}) }
\def\CBALL{ {\mathbb B}_1[{\bf 0}]}
\def\DIM{ N }
\def\HAB{ {\cal H}_A \! \otimes \! {\cal H}_B }
\def\HER{ {\rm Her} ({\cal H}, {\mathbb C}) }
\def\PSD{ {\rm Her}^{\ge 0} ({\cal H}, {\mathbb C}) }
\def\RANK{ m }
\def\RHO{ {\rm Her}^{\ge 0}_{1} ({\cal H}, {\mathbb C}) }
\newcommand{\BRA}[1] { \left\langle #1 \right|}
\newcommand{\CMATRIX}[1] {{\mathbb C}^{ #1 \! \times \! #1}}
\newcommand{\KET}[1] {\left| #1 \right\rangle}
\newtheorem{definition}{Definition}
\newtheorem{lem}{Lemma}
\newtheorem{prop}{Proposition}

\def\EFFECTIVE{AlickiJSP83,MankoMPLA92,190803699,ErdosJSP09,AlickiJSP83,190810145,190709349,Breuer2002}
\def\EXPANSIVE{MielnikJMP80,PhysRevLett.81.3992,BechmannPLA98,0502072}
\def\FRICTION{190810145,190709349,Breuer2002,GisinJPA81,GrabertZPB892,GisinJPA92,KowalskAP19,KowalskiQIP20,210308982,14057165}
\def\NORMALIZEDPTP{KrausAnnPhys71,KowalskAP19,200309170,KowalskiQIP20}
\def\NQI{9802051,9803019,9811036,0309189,BrunPRL09,BennettPRL09,13033537,13030371,13107301,150706334,200907800,220613362,XuPRR22}
\def\NQMEXPERIMENT{BollingerPRL89,PhysRevLett.64.2261,WalsworthPRL90,MajumderPRL90}
\def\NQMTHEORY{KibbleCMP78,WeinbergPRL89,WeinbergAP89,GisinPLA90,PolchinskiPRL91,CzachoPRA96,9501008,9708029,9810023,MielnikPLA01,CzachorPLA02,KentPRA05,JordanPRA06,0702171,JordanPRA10,14017018,BassiEJP15,160203772,190603869,200906422,210610576,GisinHPA89,GisinJPA95}
\def\OPEN{190810145,190709349,Breuer2002,GisinJPA81,GrabertZPB892,GisinJPA92,210308982,14057165,FonsecaRomeroPRA04,LeviPRA20,211010430}
\def\POSTSELECTION{KissPRA06,KissPRL11,GuanPRA13,GilyenSR15,KalmanPRA18,220514506}
\def\STOCHASTIC{GisinHPA89,GisinJPA92,GisinJPA95}
\def\WEAK{9905064,0008108,190712606,210708816}

\begin{document}
\rightline{\href{https://doi.org/10.1002/qute.202200156}{\tt Advanced Quantum Technologies 2200156 (2023)}}
\vspace{0.2in}
\title{Fast quantum state discrimination with nonlinear PTP channels}
\author{\textsc{Michael R. Geller}}
\affiliation{Center for Simulational Physics, University of Georgia, Athens, Georgia 30602, USA}
\date{July 4, 2023}

\begin{abstract}
\vskip 1.0in
\centerline{\bf Abstract}
\vskip 0.1in
\begin{spacing}{0.9}
We investigate models of nonlinear quantum computation based on deterministic positive trace-preserving (PTP) channels and  evolution equations. The models  are defined in any finite Hilbert space, but the main results are for dimension $\DIM \! = \! 2$. For every normalizable linear or nonlinear positive map $\phi$ on bounded linear operators $X$, there is an associated normalized PTP channel $ \phi(X) / {\rm tr}[\phi(X)]$. Normalized PTP channels include unitary mean field theories, such as the Gross-Pitaevskii equation for interacting bosons, as well as models of linear and nonlinear dissipation. They classify into 4 types, yielding 3 distinct forms of nonlinearity whose computational power we explore. In the qubit case these channels support Bloch ball torsion and other distortions studied previously, where it has been shown that such nonlinearity can be used to increase the separation between a pair of close qubit states, suggesting an exponential speedup for state discrimination. Building on this idea, we argue that this operation can be made robust to noise by using  dissipation to induce a bifurcation to a novel phase where a pair of attracting fixed points create an intrinsically fault-tolerant nonlinear state discriminator. 
\end{spacing}
\end{abstract}
\maketitle

Quantum nonlinearity, beyond the stochastic nonlinearity provided by projective measurement, might be a powerful computational resource 
\cite{\EXPANSIVE,\NQI}. However there is no experimental evidence for physics beyond standard linear quantum mechanics \cite{\NQMEXPERIMENT}. It is well known to be challenging to even formulate a consistent, fundamentally nonlinear quantum theory in accordance with general principles \cite{\NQMTHEORY}. In this paper we consider the application of {\it effective} quantum nonlinearity to information processing, while at the same time accepting that quantum physics is fundamentally linear. Effective means that it arises in some approximate (e.g., low-energy) quantum description, is a consequence of constraints on a linear system, or emerges in the limit of a large number of particles \cite{\EFFECTIVE}. An early example from particle physics is the low-energy reduction of electroweak theory to Fermi's simpler 1933 model $\mathscr{L} =  {\overline \psi} i {\slashed \partial} \psi - g \, ( {\overline \psi} \gamma_a \psi )^2 \! ,$ containing a four-fermion interaction. Effective nonlinearity is also common in condensed matter (Breuer and Petruccione \cite{Breuer2002} give an excellent introduction). It's a natural byproduct of dimensional reduction, where the dynamics of a complex quantum many-body system is described by a model with fewer degrees of freedom. A well known case is self-consistent mean field theory: For $n$ quantum particles moving in $D$ dimensions, such an effective model reduces a problem with $nD$ degrees of freedom to one involving only $D$, at the expense of nonlinearity and errors (usually). Examples include mean field models for superfluids, superconductors, and laser fields. Beyond mean field theory, various forms of nonlinearity have been proposed to describe friction and dissipation in quantum mechanics \cite{\FRICTION}, and for open systems more generally \cite{\NORMALIZEDPTP,\OPEN}.

It is not known whether effective nonlinearity can actually be used to enhance quantum information processing. Perhaps any {\it nonlinear} advantage is an artifact of approximations and could never be realized \cite{13030371,13107301,150706334}. Childs and Young \cite{150706334} used the speedup predicted for optimal qubit state discrimination with Gross-Pitaevskii nonlinearity to derive a complexity theoretic bound on the long-time accuracy of the Gross-Pitaevskii equation itself (though weaker than a known bound \cite{150706334}). Can a different nonlinearity be used, or can the Gross-Pitaevskii nonlinearity be used in another way? Is speedup the only possible advantage; what about noise resilience? The purpose of this paper is to explore these questions by providing a framework for studying known nonlinear channels from an information processing perspective, and for proposing new ones that might be experimentally realizable in the near future. The framework should be applicable to strongly correlated quantum materials, such as non-Hermitian exciton-polariton condensates \cite{GaoNat15,210701675}, that are open and operated under extreme conditions. We aim to impose as little structure as possible on the allowed evolution beyond preserving the Hermitian symmetry, positive semi-definiteness, and trace of the density matrix $X$. Stochastic nonlinearity \cite{\STOCHASTIC}, including projective measurement with postselection \cite{\POSTSELECTION} and weak continuous measurement \cite{\WEAK}, is known to be a useful resource, but is not considered here.\footnote{As discussed below, postselected measurement is formally in the NINO class, but is stochastic.}
Therefore we restrict ourselves to deterministic nonlinear positive trace-preserving (PTP) channels \cite{SudarshanPR61}. We do not attempt to derive an effective nonlinear theory from an underlying physical model, but instead try to identify what types of nonlinearity are desirable, and why. However we will only scratch the surface of this formidable problem.

Interestingly, according to work by Abrams and Lloyd \cite{PhysRevLett.81.3992}, Aaronson \cite{0502072}, and Childs and Young \cite{150706334}, nonlinearity is not required in large quantities: Even one ``nonlinear qubit'' would provide a computational benefit when coupled to a linear quantum computer. This is because nonlinearity can be used to increase the trace distance between a pair of qubit states \cite{\EXPANSIVE,150706334}, implying an exponential speedup for unstructured search and hence for any problem in the class NP.  We mainly focus on this $\DIM \! = \! 2$ case. While these results are intriguing, it should be emphasized that they apply in an idealized setting, where {\it model} errors are neglected. 

Building on a growing body of similarly motivated work \cite{\NORMALIZEDPTP,\OPEN}, we study a family of normalized PTP channels of the form $ \phi(X) / {\rm tr}[\phi(X)]$, where $\phi(X)$ is a positive linear or nonlinear map on bounded linear operators $X$ satisfying ${\rm tr}[\phi(X)] \! \neq \! 0$. Normalized PTP channels fall into 4 classes, yielding 3 distinct forms of nonlinearity. As in the classical setting, rich dynamical structures result from the interplay of nonlinearity and dissipation, and we will see that PTP channels allow for greater control over engineered {\it linear} dissipation than completely positive channels do. Our main result is the identification of a nonlinear channel where the Bloch ball separates into two basins of attraction, which can be used to implement fast intrinsically fault-tolerant state discrimination. Although we do not address the model error issue directly, we hope that the predicted {\it phase} will survive in realistic models and be observable experimentally. Section \ref{ptp channels section} mainly covers the definition of a PTP channel and can be skipped by many readers. Normalized PTP channels are classified in Sec.~\ref{ptp models section}, and fault-tolerant nonlinear state discrimination is explained in Sec.~\ref{state discrimination section}. 
The main results are summarized in  Sec.~\ref{conclusion section}.

\section{PTP channels}
\label{ptp channels section}

\subsection{Notation}

Let ${\cal H} = ( {\rm span}\{ | e_i \rangle \}_{i=1}^\DIM ,  \langle x|y\rangle )$ be the system Hilbert space with inner product $ \langle x | y \rangle = \sum_{i=1}^\DIM x_i^* y_i$, complete orthonormal basis $ \{ | e_i \rangle \}_{i=1}^\DIM  \ (\langle e_i | e_j \rangle \! = \! \delta_{ij}, \ \sum_{i=1}^\DIM e_{ii}  \! = \!  I_\DIM,  \ e_{ij} \! := \!  | e_i \rangle \langle e_j |),$ and norm $ \| |x\rangle \| =  \sqrt{\langle x | x \rangle}$. Here $x^*$ denotes complex conjugation and $I_\DIM$ is the $\DIM \! \times \! \DIM$ identity. Let $X : {\cal H} \rightarrow {\cal H} $ be a linear operator on ${\cal H}$ (isomorphic to a matrix $X \in \CMATRIX{\DIM}$) and let $X^\dagger$ be its adjoint. Also let $|X| = \sqrt{X^\dagger X}$. The set of these bounded linear operators form a second complex vector space $\BLO$; this space is our main focus. Let $\HER = \{ X \in \BLO : X = X^\dagger \} $ be the subset of self-adjoint observables, and $\PSD = \{ X \in {\rm Her}({\cal H})   : X \succeq 0 \}$  be the positive semidefinite (PSD) subset. Quantum states live in the subset of $\PSD$ with unit trace: $ \RHO = \{ X \in \PSD  : {\rm tr}(X) = 1 \}.$ In the qubit case the elements $X = (I_2 + {\bf r} \cdot \bm{\sigma})/2$ of $\RHO$ are mapped, using the basis of Pauli matrices $(\sigma^a)_{a =1,2,3}$, to real vectors ${\bf r} = (x,y,z) \in {\mathbb R}^3$ with $|{\bf r}| \le 1$, the closed Bloch ball $\CBALL$. 

\subsection{Nonlinear positive maps}

\begin{definition}[{\bf Linear map}]
Let $L: \BLO \rightarrow  \BLO $ be a map on bounded linear operators satisfying {\rm (i)} $ L(X+Y) = L(X) + L(Y)$, and  {\rm (ii)} $L(\alpha X) = \alpha L(X) $, for every $X, Y \in \BLO$ and $\alpha \in {\mathbb C}$.
Then $L$ is a {\bf linear} map on $\BLO$. 
\label{def linear}
\end{definition}

\begin{lem}[]
Let $L : \BLO \rightarrow  \BLO $ be a linear map on finite-dimensional bounded linear operators. Then $L$ has a representation
\begin{eqnarray}
X \mapsto L(X) =  \sum_{\alpha =1}^{\RANK} 
A_\alpha  X  B_\alpha, \ \ 
 A_\alpha, B_\alpha \in \CMATRIX{\DIM}, \ \ \RANK \! \le  \! \DIM^2,
 \ \ \DIM = \dim({\cal H}).
 \label{general operator-sum}
\end{eqnarray}
\label{operator-sum lemma}
\end{lem}

\begin{proof}
Every linear map is specified by its action on a complete matrix basis $e_{ab} = |e_a\rangle \langle e_b|  \in \CMATRIX{\DIM }, 
\ (e_{ab})_{a'b'} =  \delta_{aa'} \delta_{bb'}$:
\begin{eqnarray}
X \mapsto L(X) 
=  L \bigg(\sum_{a,b=1}^{\DIM} X_{ab} \, e_{ab} \bigg)
= \sum_{a,b=1}^{\DIM} X_{ab} \, L(e_{ab}).
\label{general linear map in basis}
\end{eqnarray} 
The set $\{ L(e_{ab})\}_{a,b=1}^\DIM$ defines the map. Let $\alpha =(a,b)$ be a composite index with $ a, b \in \{ 1, \dots, \DIM\}$, and rewrite (\ref{general operator-sum}) with $\RANK = \DIM^2$ as
$ L(X) = \sum_{a,b=1}^{\DIM}  A_{ab} X B_{ab} $ with 
$A_{ab}, B_{ab} \in \CMATRIX{\DIM}$. $A_{ab}$ and $B_{ab}$ can always be chosen to implement the map (\ref{general linear map in basis}): the choice $ A_{ab}= e_{ab}$ leads to $ \langle e_a | L(X) | e_d \rangle  = L(X)_{ad}  = \sum_{b,c} X_{bc} \,  (B_{ab})_{cd}$. The same matrix element of (\ref{general linear map in basis}) is $\sum_{b,c} X_{bc} \,  L(e_{bc})_{ad}$, so the choice $(B_{ab})_{cd} =  L(e_{bc})_{ad} 
$ reduces this to (\ref{general linear map in basis}) as required. $\Box$
\end{proof}

\begin{definition}[{\bf Hermitian map}]
Let  $ T : \BLO \rightarrow  \BLO $ be a map on bounded linear operators satisfying $T(X)^\dagger = T(X^\dagger) $ for every $X \in \BLO$. Then $T$ is a {\bf Hermitian} map on $\BLO$. 
\end{definition}

\noindent Hermitian maps have the property $T(\HER) \subseteq \HER$,  and they can be nonlinear.

\begin{definition}[{\bf Positive map}]
Let $\phi : \BLO \rightarrow  \BLO $ be a Hermitian map on bounded linear operators satisfying 
$\phi(  X \succeq 0) \succeq 0 $ for all $X \in \BLO$, where $ \cdot \succeq 0$ means it's an element of the PSD subset $\PSD$. Then $\phi$ is a {\bf positive} map on $\BLO$. 
\end{definition}
\noindent  We reserve the symbol $\phi$ for positive maps. 

\begin{prop}[{\bf Linear positive map} \cite{KrausAnnPhys71,ChoiLin.Alg.App.75}]
Let $\phi : \BLO \rightarrow  \BLO $ be a positive linear  map on finite-dimensional bounded linear operators. Then $\phi$ has a representation\footnote{The existence of the representation (\ref{operator-sum representation}) does not actually require $\phi$ to be positive, only Hermitian.}
\begin{eqnarray}
X \mapsto \phi(X) =  \sum_{\alpha =1}^{\RANK} \lambda_\alpha A_\alpha  X  A_\alpha ^\dagger, \ \ 
 A_\alpha \in \CMATRIX{\DIM}, \ \ 
 {\rm tr}(A^\dagger_{\alpha} A_{\beta}) = \delta_{\alpha \beta}, \ \
 \lambda_\alpha \in {\mathbb{R}},
  \ \ \RANK \! \le \! \DIM^2.
 \label{operator-sum representation}
\end{eqnarray}
 \label{operator-sum representation prop}
\end{prop}

\begin{proof}
Define a Choi operator 
$C := \sum_{i,j =1}^\DIM \phi(e_{ij}) \otimes e_{ij}$ in an expanded Hilbert space $\HAB$ consisting of two copies of our system ${\cal H}$, with ${\dim(\cal H}_{A,B})  \! = \! \DIM .$ The matrix elements of $C$ in the product basis $ \{ |e_{a} \rangle \otimes |e_{b} \rangle \}_{a,b=1}^\DIM $
are $(\langle e_{a} | \otimes \langle e_{b} |) \, C \, 
(|e_{a'} \rangle \otimes |e_{b'} \rangle ) 
= \langle e_a | \phi(e_{bb'}) | e_{a'} \rangle
= \phi(e_{bb'})_{aa'}$. 
Using the Hermitian property $\phi (X)^\dagger = \phi (X^\dagger)$ we see that $C \in \CMATRIX{\DIM^2}$ is Hermitian:
$ (\langle e_{a} | \otimes \langle e_{b} |) \, C^\dagger \, 
(|e_{a'} \rangle \otimes |e_{b'} \rangle ) 
= \big( \langle e_{a'} | \otimes \langle e_{b'} | \, C \, 
| e_{a} \rangle \otimes |e_{b} \rangle \big)^*
=  \phi(e_{b'b})_{a'a}^* =  \phi(e_{bb'})_{aa'}
= (\langle e_{a} | \otimes \langle e_{b} |) \, C \, 
(|e_{a'} \rangle \otimes |e_{b'} \rangle )
$. It therefore has a spectral decomposition
$C = \sum_{\alpha=1}^{\DIM^2} \lambda_\alpha \, | V_\alpha \rangle \langle V_\alpha |, \ \lambda_\alpha \in {\mathbb R},$ where the eigenvectors $ \{ | V_\alpha \rangle \}_{\alpha=1}^{\DIM^2}$ form a complete orthonormal basis in $\HAB$.  Any $| \psi \rangle \in \HAB $ can be obtained by the application of an operator $\DIM^{\frac{1}{2}} \, A \otimes I_\DIM $ to $| \Psi \rangle := \DIM^{-\frac{1}{2}}  \sum_{i=1}^\DIM |e_i\rangle \otimes |e_i \rangle $,
with $A_{ij} = \langle e_i | A | e_{j} \rangle =  \,(\langle e_{i} | \otimes \langle e_{j} | )\,  | \psi \rangle $. Then
$| V_\alpha \rangle \langle V_\alpha | = N \, ( A_\alpha \otimes I) | \Psi \rangle \langle  \Psi  | 
( A^\dagger_\alpha \otimes I)$ where
$\langle e_i | A_\alpha | e_{j} \rangle = (\langle e_{i} | \otimes \langle e_{j} | )\,  |  V_\alpha \rangle $. Also note that $| \Psi \rangle \langle  \Psi  |  = \DIM^{-1}  \sum_{i,j=1}^\DIM e_{ij} \otimes e_{ij}$, so
\begin{eqnarray}
C &=& N \sum_{\alpha=1}^{\DIM^2} \lambda_\alpha \,
( A_\alpha \otimes I) | \Psi \rangle \langle  \Psi  | 
( A^\dagger_\alpha \otimes I)
 = \sum_{i,j=1}^\DIM  \bigg(  \sum_{\alpha=1}^{\DIM^2} \lambda_\alpha   A_\alpha e_{ij} A_\alpha^\dagger \bigg) \otimes e_{ij} .
 \end{eqnarray}
Then $\phi(e_{ij}) =  \sum_{\alpha=1}^{\DIM^2} \lambda_\alpha   A_\alpha e_{ij} A_\alpha^\dagger$, and by linearity and completeness, 
$\phi(X) =  \sum_{\alpha=1}^{\DIM^2} \lambda_\alpha   A_\alpha X A_\alpha^\dagger$, as required. Furthermore, because the $|V_\alpha \rangle$ are orthonormal, $\langle V_\alpha |  V_\beta \rangle = 
{\rm tr}(A^\dagger_{\alpha} A_{\beta})  =  \delta_{\alpha \beta}. \ \Box$
\end{proof}

Having established a general representation for positive linear maps, we give some nonlinear examples. The first is
$ X \mapsto \phi (X) = 
(A^\dagger + CXB^\dagger) (A + BXC^\dagger) = [\phi (X^\dagger)]^\dagger , 
\ \ A, B, C \in \CMATRIX{\DIM}. $
On any PSD input, $\phi (X)  = | A + BXC^\dagger |^2 \succeq 0$, so $\phi$ is positive. Two examples of discrete nonlinear positive maps are $ X \mapsto  \phi_\pm(X) =  ( {\rm tr}(X) I_\DIM \pm X )^2 =  [\phi_\pm (X^\dagger)]^\dagger, $ which reduce to $| {\rm tr}(X) I_\DIM \pm X |^2 \succeq 0$ on PSD inputs. Another example is $ X \mapsto \phi (X) = |{\rm det}(X)|  \cdot I_\DIM = [\phi (X^\dagger)]^\dagger, $ which (after normalization) maps every input to the  infinite-temperature thermal state. Permanents and determinants have also been used to construct nonlinear completely positive trace-preserving (CPTP) channels \cite{AndoChoi1986}.  

\begin{definition}[{\bf PTP channel} \cite{SudarshanPR61}]
Let $\Lambda :  \RHO \rightarrow  \RHO $  be a positive map satisfying $ {\rm tr}(\Lambda (X)) = 1$ for all $X \in \RHO $. Then $\Lambda$ is a {\bf positive trace-preserving} {\rm (PTP)}  channel.
\label{ptp channel def}
\end{definition}
\noindent 
In this paper $\Lambda$ always refers to a PTP channel.
The PTP channels defined here are endomorphisms on the state space $\RHO$, and are only required to act properly on physical inputs. They may be composed of maps defined on operators outside of $\RHO$ as well. In this case, the behavior of the extended PTP channels on inputs outside of $\RHO$ is not constrained by {Definition \ref{ptp channel def}}.\footnote{In particular, they might not be trace preserving in the usual sense of ${\rm tr}(\Lambda (X)) \! = \!  {\rm tr}(X)$, because the trace preservation in  {Definition \ref{ptp channel def}} is really a trace {\it fixing} condition.} 

\section{PTP Models}
\label{ptp models section}

In this paper we investigate models of nonlinear quantum computation based on a category of PTP channels that take the form of a nonlinear positive channel rescaled to conserve trace \cite{\NORMALIZEDPTP}. Although normalized PTP channels do not provide an exhaustive classification of nonlinear channels, they are sufficient to illustrate some of the different types of effective nonlinearity that are available or might become available in the near future. 

\begin{definition}[{\bf Normalized PTP channel} \cite{\NORMALIZEDPTP}]
Let $\phi : \RHO \rightarrow \RHO $ be a positive map satisfying ${\rm tr}[\phi(X)] \neq 0$ for all $X \in \RHO $. Then the PTP map 
\begin{eqnarray}
\Lambda_\phi  : \RHO \rightarrow \RHO
\ \ {\rm given \ by} \ \ 
X \mapsto \Lambda_\phi (X)= \frac{\phi(X)}{{\rm tr}[\phi(X)]}
\label{normalized map def}
\end{eqnarray}
is a {\bf normalized PTP channel}.
\label{scaled PTP}
\end{definition}
\noindent 
Normalized PTP channels are common because they can be constructed from any positive map $\phi$ satisfying a normalizability condition ${\rm tr}[\phi(X)] \neq 0$.\footnote{To be precise, normalized PTP channels are only defined on a subset of physical states  $\RHO$, because their action on inputs with ${\rm tr}[\phi(X)] = 0$ is not specified. However this situation does not arise in the cases considered here.}
From positivity alone, ${\rm tr}[\phi(X)] \ge 0$.
Normalizability requires the more restrictive condition that ${\rm tr}[\phi(X)]  > 0$. The condition ${\rm tr}[\phi(X)]  > 0$ is important in (\ref{normalized map def}) because it ensures the positivity of $\Lambda_\phi $. Normalized PTP channels naturally fall into four classes, according to whether $\phi$ is linear (or not) and trace-preserving (or not):
\begin{enumerate}

\item[A.] Linear positive $\phi$ and ${\rm tr}[\phi(X)]=1 \ {\rm for \  all} \ X \in \RHO$. This is the class of general linear positive channels, which includes linear CPTP channels.

\item[B.]  Linear positive $\phi$ and ${\rm tr}[\phi(X)] \neq 1 \ {\rm for \ some} \ X \in \RHO$. These can be called nonlinear in normalization only (NINO) channels.

\item[C.]  Nonlinear positive $\phi$ and ${\rm tr}[\phi(X)] = 1 \ {\rm for \  all} \ X \in \RHO$. This class includes {\it state-dependent} CPTP channels.

\item[D.] Nonlinear positive $\phi$ and ${\rm tr}[\phi(X)] \neq 1 \ {\rm for \  some} \ X \in \RHO$. These are the most general channels considered here. They support rich dynamics similar to that of classical nonlinear systems.

\end{enumerate}
These classes are discussed below in Secs.~\ref{linear ptp and cptp section} - \ref{general normalized ptp section}. 

\subsection{Linear PTP and CPTP channels}
\label{linear ptp and cptp section}
 
A standard form for linear PTP channels is provided in  {\bf Proposition \ref{operator-sum representation prop}}.  While every linear positive map can be put in the form
(\ref{operator-sum representation}), not every map of the form (\ref{operator-sum representation}) is positive (because the $\lambda_\alpha$ can be negative). The necessary conditions for (\ref{operator-sum representation}) to represent a positive (P) or completely positive (CP) map can be obtained as follows: Let $S_{>} \! =  \! \{ \alpha \! : \! \lambda_\alpha \!  > \! 0 \} \neq \emptyset $ and $S_{<} \! = \! \{ \alpha \! : \! \lambda_\alpha \! <  \! 0 \}$ be the index sets of positive and negative Choi eigenvalues, and decompose $\phi$ into $\phi  = \phi_{>} -  \phi_{<} $, where $\phi_{>}(X) =  \sum_{\alpha \in S_{>}} \lambda_\alpha   \, A_\alpha X A_\alpha^\dagger$ and $\phi_{<}(X) =  \sum_{\alpha \in S_{<}} |\lambda_\alpha|  \,  A_\alpha X A_\alpha^\dagger$ are each manifestly positive. Upon rescaling, each can be put into the form $\Phi(X) = \sum_\alpha A_\alpha X A^\dagger_\alpha , \ 
A_\alpha \in \CMATRIX{\DIM}$. Maps of this form also satisfy the stronger condition of complete positivity, meaning that they remain positive when combined with a second Hilbert space ${\cal H}_B$, of any finite dimension, on which the identity acts:
$[\Phi \otimes  {\rm id}_B](X\succeq 0)  \succeq 0$.
Here $X \in B({\cal H}_A  \otimes {\cal H}_B , {\mathbb C})$. Each term in $\Phi$ clearly has the required property:
$\BRA{\psi} (A_\alpha \otimes I ) X (A^\dagger_\alpha \otimes I) \KET{\psi}  = \BRA{\psi_\alpha} X \KET{\psi_\alpha} \ge 0$ 
for every $\KET{\psi} \in \HAB $,
where $ \KET{\psi_\alpha}  = (A^\dagger_\alpha \otimes I)  \KET{\psi}$. The condition for positivity is therefore $ \phi_{>}(X) \succeq  \phi_{<}(X) \  {\rm for \ every} \  X \succeq 0 $, whereas the condition for complete positivity is $\phi_{<}(X) = 0$. An important CP map is the completely-positive trace-preserving (CPTP) channel, which has a nonnegative Choi spectrum $\lambda_\alpha \ge 0$ and therefore an operator-sum representation \cite{KrausAnnPhys71,ChoiLin.Alg.App.75}
\begin{eqnarray}
X \mapsto \Phi(X) = \sum_{\alpha =1}^{\RANK} A_\alpha  X  A_\alpha ^\dagger, \ \ 
\sum_{\alpha  =1}^{\RANK} A_\alpha^\dagger A_\alpha  = I_\DIM , 
\ \  A_\alpha  \in \CMATRIX{\DIM}, \ \ 
\RANK \! \le \! \DIM^2.
 \label{operator-sum representation cptp}
\end{eqnarray}
This follows from {\bf Proposition \ref{operator-sum representation prop}} (but with different $A^{\rm ' s}$ 
as to enforce trace conservation and absorb a factor of $\sqrt{\lambda_\alpha}$).
The integer $\RANK$ is the {\bf Choi rank}. In this paper $\Phi$ always refers to a linear CP (and usually TP) map.

There is a large body of work on linear P channels that are not CP, called {\bf non-CP} maps \cite{PechukasPRL94,ShajiPLA05,CarteretPRA08,13120908,150305342}. The distinction between CP and non-CP maps is the possibility of non-CP maps generating negative states on entangled inputs, requiring their restriction to a subset of the state space where this unphysical output is avoided. Although non-CP operations are well known for their use in entanglement detection \cite{PeresPRL96}, the question of whether linear non-CP channels could provide a {\it computational} advantage over linear CP channels appears to be largely unexplored \cite{150305342}. However, non-CP communication channels in which the environment is measured (and the results used to correct the channel) enable increased capacity \cite{0403092,0409026,0507045}.  In the presence of nonlinearity, we also find that the distinction between CP and non-CP is significant, due to the larger space of nonunitary processes supported by non-CP channels.

\subsection{NINO channels}
\label{nino section}
 
\begin{definition}[{\bf NINO channel}]
Let $\phi :  \RHO \rightarrow  \RHO  $ be a positive linear map with ${\rm tr}[\phi(X)] \neq 0 $ for every $X \in \RHO $, and ${\rm tr}[\phi(X)] \neq 1 $ for one or more $X \in \RHO$. Then the PTP map 
\begin{eqnarray}
\Lambda_\phi :  \RHO \rightarrow \RHO \ \ {\rm given \ by} \ \ 
X \mapsto \Lambda_\phi (X) = \frac{\phi(X)}{{\rm tr}[\phi(X)]} 
\label{nino def equation}
\end{eqnarray}
is called a {\bf nonlinear-in-normalization only} (NINO) channel. 
\label{nino def}
\end{definition}
\noindent 
We stress that the positive map $\phi$ in  Definition \ref{nino def} is {\it linear}, but not unitary $\phi_U(X) = UXU^\dagger, \ U^\dagger U \! = \!  I_\DIM$ (because ${\rm tr}[\phi_U(X)] = {\rm tr}(X) = 1$  for all $X$). NINO channels inherit, from linear maps, the powerful ability to characterize them through their action on a complete basis (because $\phi$ has this property). Next, extending Rembieli\'nski and Caban \cite{200309170},  we obtain a general representation and Markovian  evolution equation for NINO channels.

\begin{prop}[{\bf NINO representation}]
Let $\Lambda$ be a NINO map of the form (\ref{nino def equation}) with $\phi$ linear. Then $\Lambda$ has a representation
\begin{equation}
X \mapsto \Lambda(X) = \frac{\displaystyle \sum_{\alpha =1}^{\RANK} \zeta_\alpha A_\alpha  X  A_\alpha ^\dagger}{ {\rm tr} (FX)}, \ \  F := \sum_{\alpha =1} ^{\RANK} \zeta_\alpha A_\alpha ^\dagger A_\alpha  \neq I_\DIM , \ \ 
 A_\alpha  \in \CMATRIX{\DIM}, 
 \ \ \zeta_\alpha = \pm 1, \ \ 
\RANK \! \le \! \DIM^2.
 \ \ \ \ 
 \label{nino representation prop equation}
\end{equation}
 \label{nino representation prop}
 \end{prop}

\begin{proof}
This follows from {\bf Proposition \ref{operator-sum representation prop}} after substituting 
$\lambda_\alpha = \zeta_\alpha | \lambda_\alpha |, \ \zeta_\alpha = {\rm sign}( \lambda_\alpha)$, and rescaling the $A^{\rm 's}$. $ \ \Box$
\end{proof}

Our main objective is to study NINO channels from a dynamical perspective, via evolution equations. Much of the following analysis will carry over to the other classes as well. Let $X \in \RHO$ be the state of a physical system which evolves continuously in time according to
\begin{eqnarray}
X(t) \mapsto X(t + \Delta t)  = \Lambda_{\Delta t, t} \big( X(t) \big) , \ \ 
\Delta t  \ge  0, \ \ 
\Lambda_{0, t}  = {\rm id} \ {\rm for \ all} \ t \in {\mathbb R}.
\label{def continuous dynamics}
\end{eqnarray}
Here $\Lambda_{\Delta t, t}$ is a two-parameter family of PTP channels continuous in both $t$ and $\Delta t$. Because we want to study the simplest instances that illustrate computational advantages, we make several additional simplifying assumptions:
\begin{enumerate}

\item[(i)]  Stationarity: $\Lambda_{\Delta t, t} = \Lambda_{\Delta t}$ for all $t \in  {\mathbb R}$.

\item[(ii)] Semigroup: 
$ \Lambda_{s} \circ \Lambda_{t} =  \Lambda_{s+t},  \ \ 0 \le s \ll 1, \ \ 0 \le t \ll 1.$

\item[(iii)]  The nonlinearity can be turned off, recovering linear CPTP evolution.

\item[(iv)] $\DIM=2$.

\end{enumerate}
From (\ref{def continuous dynamics}) we have that $\Lambda_{\Delta t}$, defined in (i), is continuous and satisfies $\Lambda_{0} = {\rm id}$. Stationarity excludes time-dependent Hamiltonians that arise when a physical system is driven with time-dependent fields (precisely what we want to do to run a device). While it is possible to formulate the  problem with time-dependent generators and obtain some of the results in terms of time-ordered exponentials, we will not cover that case here. Instead we assume that the strength of the nonlinearity, $g$, can be turned on and off, and while in the off state the full toolkit of linear quantum information processing can be applied. The  additional assumptions (i) and (ii) are sufficient to define, for any PTP channel, a Markovian evolution equation that extends linear Markovian CPTP evolution by the Gorini-Kossakowski-Sudarshan-Lindblad  (GKSL) equation \cite{GoriniJMP76,LindbladCMP76}. Our approach  follows a large body of work on nonlinear evolution equations \cite{\EFFECTIVE,\NORMALIZEDPTP,\OPEN}. Especially relevant are the recent papers by Kowalski and Rembieli\'nski \cite{KowalskAP19}, Fernengel and Drossel \cite{190709349}, and Rembieli\'nski and Caban \cite{200309170}. In contrast with \cite{200309170}, we do not assume a semigroup property at long times, which can be violated in the presence of initial system-environment correlation \cite{13120908,0611057}. The restriction to qubits is not used in the derivation of the evolution equations, but is needed for their subsequent analysis.

Before deriving an evolution equation for (\ref{nino representation prop equation}), we briefly consider the $\RANK=1$ case to illustrate one of the differences between NINO and linear CPTP channels. Suppose $A = e^{tL}$ is a continuous time-dependent linear operator  infinitesimally generated by some $L \in B({\cal H})$. Decompose the generator into Hermitian and  anti-Hermitian contributions $L = L_{+} + L_{-} $, with $L_{\pm} := (L \pm L^\dagger)/2$. Then  $X(0) \mapsto X(t) = \Lambda_t(X(0)) 
= (e^{tL} \, X(0)  \, e^{tL^\dagger}) / [{\rm tr} [e^{tL^\dagger} e^{tL}  X(0)] ] $
and
\begin{eqnarray}
\frac{dX}{dt} = [L_- , X] + \{ L_+ , X \} - 2 g \,  {\rm tr} (L_+ X) \,  X, 
 \label{nino example evolution}
\end{eqnarray}
with $g=1$. This evolution equation has been discussed previously by several authors in different contexts \cite{BrodyPRL12,14057165,KowalskAP19,200309170}. Here $\{ \cdot , \cdot \}$ is an anticommutator.
Tracing gives $ \frac{d}{dt} {\rm tr}(X) 
 \! = \!  2 \, {\rm tr} (L_{+} X)  \!  -   \!  2 g \, {\rm tr} (L_{+} X)  \, {\rm tr}(X)  \!  =  \!  0$ assuming ${\rm tr}(X) = 1$, showing how the nonlinear term fixes the normalization. The equation of motion (\ref{nino example evolution}) includes unitary evolution generated by a Hamiltonian $H= i L_{-}$, together with linear dissipation and amplification (if $L_{+}$ has positive eigenvalues). Or we can say that the dynamics is generated by a {\it non-Hermitian} Hamiltonian\footnote{We distinguish between infinitesimal generators $G$ and Hamiltonians $iG$. Thus, unitary evolution results from Hermitian Hamiltonians but from anti-Hermitian generators.}
$ H_{\rm non} := iL = H + iL_{+} $ \cite{14057165}. NINO channels expand the utilization of linear maps by conserving the trace nonlinearly.

Now we obtain an evolution equation for general $m$. Continuous one-parameter NINO channels have the form
\begin{eqnarray}
\Lambda_t(X) = \frac{\displaystyle \sum_{\alpha=1}^\RANK   \zeta_\alpha \,  A_\alpha(t) \, X \, A_\alpha^\dagger(t)}{  {\rm tr}(F_t X)},
\ \  F_t := \sum_{\alpha=1}^\RANK  \zeta_\alpha \,
A_\alpha^\dagger(t) \, A_\alpha(t) \neq I_\DIM,
\ \ \zeta_\alpha = \pm 1 , \ \ t \ge 0.
\label{def continuous NINO}
\end{eqnarray}
The operators $A_\alpha(t) \in \CMATRIX{\DIM}$ are analytic functions of time for $t \ge 0$, which can be classified into two types, jump and nonjump, according to their $t \rightarrow 0$ behavior. We assume that $k>0$ of the $\RANK$ operators have {\it nonzero} limits
\begin{eqnarray}
\lim_{t \rightarrow 0}  A_\alpha(t) = A_\alpha^0 \neq 0, \ \ \alpha \in \{1, \cdots, k \} 
\ \ {\rm (nonjump)},
\label{def nonjump}
\end{eqnarray}
while the others vanish
\begin{eqnarray}
\lim_{t \rightarrow 0}  A_\alpha(t) = 0  , \ \ \alpha \in \{k \! + \! 1, \cdots, \RANK \} 
\ \ {\rm (jump)} .
\label{def jump}
\end{eqnarray}
The jump/nonjump terminology comes from the unraveled stochastic picture, where one or more rare but disruptive ``jump" operations are applied randomly to a simulated open system, on top of a smooth background of unitary plus nonunitary evolution, with the nonunitary component fixed to conserve trace.  There is no restriction on the number of jump operators; however no more than $\DIM^2$ are required. The only restriction on the number of nonjump operators is that $k > 0$:  At least one is required to obtain the desired $t \rightarrow 0$ limit. Note that if $F_t = I_\DIM$, the nonlinear generator disappears and we recover linear GKSL evolution \cite{GoriniJMP76,LindbladCMP76}.

The constant matrices $A_\alpha^0 \in \CMATRIX{\DIM} $  in (\ref{def nonjump}) are not arbitrary; the condition
$ \lim_{t \rightarrow 0}  \Lambda_t(X) = X $
for all $X \in \RHO$ requires 
$A_\alpha^0 = z_\alpha \, I_\DIM, \ \ \alpha \in \{1, \cdots, k\},$ where the $z_\alpha \in {\mathbb C}$ satisfy a ``normalization'' condition 
$ \sum_{\alpha=1}^k \zeta_\alpha |z_\alpha|^2 = 1$ (recall that $\zeta_\alpha = \pm 1$). To satisfy the semigroup property it is sufficient to let
\begin{eqnarray}
A_{\alpha}(t) = 
\begin{cases}
z_\alpha \, e^{tL_{\alpha}}  \ \ {\rm for} \  \alpha \in \{1, \cdots, k\},  \\
B_{\alpha} \sqrt{t}  \ \ \ {\rm for} \  \alpha \in \{k \! + \! 1, \cdots, \RANK \},
\end{cases}
\label{def markovian model}
\end{eqnarray}
where $L_{\alpha}, B_{\alpha} \in \BLO$. To see why, note that in the short-time limit
\begin{eqnarray}
 \sum_{\alpha=1}^\RANK \zeta_\alpha \, A_\alpha(t) X  A_\alpha^\dagger(t)
 = X + t  \bigg[  \sum_{\alpha=1}^{k} \zeta_\alpha |z_\alpha |^2
 \bigg( \! L_{\alpha} X + X L_{\alpha}^\dagger \! \bigg) + \sum_{\alpha > k}^{\RANK} \zeta_\alpha B_\alpha X B_\alpha^\dagger
 \bigg] + O(t^2)
\end{eqnarray}
and
\begin{eqnarray}
 F_t &=& \sum_{\alpha=1}^\RANK  \zeta_\alpha \, A_\alpha^\dagger(t) A_\alpha(t),  \\
 &=&  I_\DIM + t  \bigg[  \sum_{\alpha=1}^{k} \zeta_\alpha |z_\alpha |^2
 \big( \! L_{\alpha} + L_{\alpha}^\dagger \! \big) + \sum_{\alpha > k}^{\RANK} \zeta_\alpha B_\alpha^\dagger B_\alpha 
 \bigg] + O(t^2), \\
 &=&  I_\DIM + t \, \frac{dF_0}{dt} + O(t^2).
\end{eqnarray}
Then, if ${\rm tr}(X) = 1$,
\begin{eqnarray}
 \Lambda_t(X) &=& X + t  \bigg[  \sum_{\alpha=1}^{k} \zeta_\alpha |z_\alpha |^2
 \bigg( \! L_{\alpha} X + X L_{\alpha}^\dagger \! \bigg) + \sum_{\alpha > k}^{\RANK} \zeta_\alpha B_\alpha X B_\alpha^\dagger  
 - {\rm tr} \bigg( \! X \frac{dF_0}{dt} \bigg) X \bigg]
 + O(t^2) \ \ \ \ \ 
 \label{short time NINO evolution}
 \end{eqnarray}
and
\begin{eqnarray}
  \Lambda_s(\Lambda_t(X)) = \Lambda_{s+t}(X)
 + O(s^2) + O(st) + O(t^2)
 \ \ {\rm for \ every \ } X \in \RHO,
\end{eqnarray}
as required for short times. The NINO evolution equation  follows from (\ref{short time NINO evolution}):
\begin{eqnarray}
\frac{dX}{dt} &=&  \sum_{\alpha=1}^{k} \zeta_\alpha  |z_\alpha |^2
 \bigg( \! L_{\alpha} X + X L_{\alpha}^\dagger \! \bigg) + \sum_{\alpha > k}^{\RANK} \zeta_\alpha B_\alpha X B_\alpha^\dagger  
 - {\rm tr} \bigg( \! \! X \frac{dF_0}{dt} \! \bigg) \, X
  \nonumber  \\
 &=&  \sum_{\alpha=1}^{k} \zeta_\alpha |z_\alpha |^2
 \bigg( \! [L_{\alpha-} , X] + \{ L_{\alpha+} , X \} \! \bigg) + \sum_{\alpha > k}^{\RANK} \zeta_\alpha B_\alpha X B_\alpha^\dagger  
 - g \, {\rm tr} \bigg( \! \! X \frac{dF_0}{dt} \!  \bigg) \, X, 
\label{def NINO evolution}
\end{eqnarray}
where, in the second line, each linear operator $L_\alpha$ has been decomposed into Hermitian and anti-Hermitian parts according to $L_\alpha = L_{\alpha+} + L_{\alpha-} $, with $L_{\alpha\pm} := (L_\alpha \pm L_\alpha^\dagger)/2$. For future reference we also introduced a nonlinear coupling strength $g=1$. The anti-Hermitian $ \{ L_{\alpha-} \}_{\alpha=1}^k $ each generate a unitary time evolution with Hamiltonian $H_\alpha = i L_{\alpha-} = H_\alpha^\dagger$, whereas the $ \{ L_{\alpha+} \}_{\alpha=1}^k $  and $ \{ B_{\alpha} \}_{\alpha>k}^\RANK $  generate nonunitary time evolution.
The brackets in (\ref{def NINO evolution}) are commutators and anticommutators 
$\{A,B\} = AB+BA$. Note that the parameters $z_\alpha$ can be absorbed into rescaled generators
$ | z_\alpha|^2 \, L_\alpha \rightarrow L_\alpha $ with no essential change.  After doing this we obtain, for the common case of a single nonjump operator ($k=1$),
\begin{eqnarray}
\frac{dX}{dt} = [L_{-}, X] + \{ L_{+}, X \}  +  \sum_\alpha \zeta_\alpha  B_\alpha X B_\alpha^\dagger +   g \, {\rm tr}(X\Omega) X, \label{def NINO evolution k=1}
 \end{eqnarray}
 where
 \begin{eqnarray}
\Omega := - 2 L_{+}   -    \sum_\alpha \zeta_\alpha \, B_\alpha^\dagger B_\alpha .
\label{omega definition}
 \end{eqnarray}
To the best of our knowledge, the NINO evolution equation (\ref{nino example evolution}) without jump operators was first obtained by Brody and Graefe \cite{BrodyPRL12}, and (\ref{def NINO evolution k=1}) was first obtained by Zloshchastiev and Sergi \cite{14057165}. Trace dynamics, according to  (\ref{def NINO evolution k=1}), satisfies 
\begin{eqnarray}
 \frac{d \tau}{dt} = ( g \tau \! - \! 1) \, {\rm tr}(X \Omega), \ \ \tau := {\rm tr}(X).
\label{nino trace equation}
 \end{eqnarray}
In the $g=0$ linear case, trace is conserved by requiring $\Omega=0$, whereas in the nonlinear case trace conservation requres $g = 1/\tau = 1$.
 
A principal difference between Markovian NINO and Markovian CPTP evolution is in the form of the  nonjump operators. In a linear CPTP channel, $\zeta_\alpha = 1 $ and trace is conserved through the requirement
$ L_{+} = - \frac{1}{2} 
\sum_\alpha B_\alpha^\dagger  B_\alpha 
$, which sets $\Omega = 0$ in (\ref{def NINO evolution k=1}) and (\ref{nino trace equation}).
In this case the total Hermitian generator $L_{+}$ is always negative semidefinite, leading to nonexpansive evolution and usually to a single stable fixed point. However in a NINO channel the $ L_{\alpha+} $ are free parameters, and they can have positive eigenvalues. An example of this distinction occurs when $m \! = \! 1$: Rank 1 CPTP channels are unitary and nondissipative, whereas $m \! = \! 1$ NINO channels already support dissipation and amplification [recall (\ref{nino example evolution})]. Therefore we can think of the NINO evolution equations (\ref{def NINO evolution}) and  (\ref{def NINO evolution k=1}) as generalizations of the linear GKSL equation to support linear evolution by one or more non-Hermitian Hamiltonians \cite{14057165}.

A well known NINO channel is projective measurement followed by postselection \cite{Kraus1983}. In this case the positive linear map in {Definition \ref{nino def}} is $\phi(X) =  A X A^\dagger$ where $A \in \CMATRIX{\DIM}$ is the selected measurement operator. However postselected measurement does not implement a 
deterministic channel and using it comes with an overhead determined by the frequency of rejected measurement outcomes. Furthermore, postselected measurement does not implement the continuous-time evolution equation (\ref{def NINO evolution k=1}).  For these reasons we do not consider postselected measurement channels in this paper. 

\subsection{State-dependent CPTP channels}
\label{state-dependent cptp section}

Next we discuss the class of normalized PTP channels (\ref{normalized map def}) with nonlinear positive $\phi$ and ${\rm tr}[\phi(X)] = 1 \ {\rm for \  all} \ X \in \RHO$. These include the important subset of parametrically nonlinear CPTP channels, which can be called state-dependent CPTP channels.

\begin{definition}[{\bf State-dependent CPTP}]
Let $X \in \BLO$ and $ A_\alpha(X) \in \CMATRIX{\DIM} $ be a set of $X$-dependent matrices satisfying
\begin{enumerate}
\item $  A_\alpha(X^\dagger) = A_\alpha(X),$
\item ${\displaystyle \sum_{\alpha =1}^{\RANK} A_\alpha(X)^\dagger A_\alpha(X) = I_\DIM} $,
\end{enumerate}
for all $X \in \BLO$ and any finite $\RANK$. Then 
\begin{eqnarray}
X \mapsto \Lambda(X) = \sum_{\alpha =1}^{\RANK} A_\alpha(X) \, X \, A_\alpha(X)^\dagger
\end{eqnarray}
is a {\bf state-dependent CPTP} channel.
\end{definition}

\noindent Channels in this class have been investigated by many authors \cite{\EXPANSIVE,13033537,13030371,13107301,150706334,\NQMTHEORY,\EFFECTIVE,GisinJPA81}.
The associated evolution equation is the state-dependent GKSL equation. Many early proposals for nonlinear extensions of quantum mechanics, including the Weinberg model \cite{WeinbergAP89}, and unitary models based on a nonlinear Schr\"odinger equation, are in this class. A rank 1 example is 
\begin{eqnarray}
X \mapsto \Lambda(X) = U(X) \, X \, U(X)^\dagger, \ \ 
U(X) := e^{i \, {\rm tr}(AX) \, B} = U(X^\dagger) , \ \ 
A, B \in \HER.
\label{rank 1 state dependent cptp equation}
\end{eqnarray}
This map applies a generator $B$ scaled by the mean $\langle A \rangle = {\rm tr}(AX)$ of observable $A$. A generalization of
(\ref{rank 1 state dependent cptp equation}) to multiple nonlinear generators is
$U(X) = e^{i \sum_\alpha {\rm tr}(A_\alpha X) \, B_\alpha}$, which includes arbitrary state-dependent Hamiltonians
and unitary mean field theories, including the Gross-Pitaevskii equation for interacting bosons.

For a qubit in the Pauli basis, $X = (I_2 + {\bf r} \cdot \bm{\sigma})/2 \in \RHO$, ${\bf r} = {\rm tr}( X {\bm \sigma}) \in \CBALL $, any Markovian PTP evolution equation can be put in the form\footnote{Let $dX/dt \! = \! Y(X) \! = \! \xi^a(X) \sigma^a =  \xi^a({\bf r}) \sigma^a$ where each $\xi^a :  {\mathbb R}^3 \rightarrow {\mathbb R} $ is a continuous function of ${\bf r}$, which can be decomposed as $\xi^a({\bf r}) = 2 C^\alpha + 2 G^{ab}({\bf r}) r^b = 2 C^\alpha + 2 L^{a b} r^b + 2 g N^{a b}({\bf r}) r^b$, where $2 C^\alpha$ captures any ${\bf r}$-independent part of $\xi^a({\bf r})$.}
\begin{eqnarray}
\frac{dX}{dt} = \frac{\sigma^a }{2}  
\bigg( \frac{dr^a}{dt} \bigg), \ \ 
\frac{dr^a}{dt}  = {\rm tr} \bigg( \! \frac{dX}{dt} \sigma^a \!\bigg)  \!  = G^{a b}({\bf r}) \, r^b + C^a,
\ \ G({\bf r}) \in {\mathbb R}^{3 \times 3}, \ \ 
C^a \in {\mathbb R}^3,
\label{general qubit evolution equation}
\end{eqnarray}
where we sum over repeated indices $a, b \in (1,2,3)$. Here  $G({\bf r})$ is a state-dependent generator, which can be decomposed into linear and nonlinear parts: $G^{a b}({\bf r}) = L^{a b} + g N^{a b}({\bf r})$, with nonzero coupling $g$ indicating the presence of nonlinearity. If $N^{a b}({\bf r}) \! = \! 0$, (\ref{general qubit evolution equation}) describes a general affine transformation on $X$ and ${\bf r}$ (strictly linear if $C^a \! = \! 0$). Every $G({\bf r})$ can be  decomposed into symmetric and antisymmetric components $G = G_{+} +G_{-}$, with $G_{\pm} := (G \pm G^\top)/2$, which have distinct actions on the Bloch vector length: $\frac{d}{dt}  | {\bf r} |^2 = 2 G^{a b}({\bf r})  \, r^a r^b + 2 r^a C^a =2 G^{a b}_{+}({\bf r})  \, r^a r^b + 2 r^a C^a$. Antisymmetric components $G_{-}$ conserve Bloch vector length; they result from (possibly state-dependent) ``unitary'' transformations $X \mapsto U(X) \, X \, U(X)^\dagger$. {\it Linear} antisymmetric generators correspond to rigid rotations of the Bloch ball and result from strictly linear unitary transformations on $X$. General symmetric generators $G_{+}({\bf r})$ can amplify some qubit states, increasing their Bloch vector, while decreasing others. A process that increases (decreases) $| {\bf r}|$ is called amplifying (dissipative), and amplification is entropy decreasing. Linear symmetric generators $G_{+}$ resulting from CPTP channels have nonpositive $G_{+}$ [see discussion following (\ref{def NINO evolution k=1})] but can increase $ | {\bf r} |$ if the channel is nonunital ($C^a \neq 0$). Thus, Bloch vector amplification does not immediately imply nonlinearity of the evolution equation.
 
Instead we consider a geometric characterization of the dynamics that is specifically sensitive to the presence of a nonlinear or non-CP map: The divergence of the qubit velocity field is
\begin{eqnarray}
\nabla \cdot  (d {\bf r} /dt) = {\rm tr}[G_{\! +}({\bf r})] + g \, r^b \partial_a  G^{a b}({\bf r}),
\label{qubit flow divergence equation}
\end{eqnarray}
where $g=1$, which has contributions from both linear symmetric and nonlinear generators, and can take either sign. By contrast, the divergence is nonpositive in linear CPTP channels (because $G_{+} \preceq 0$). 
So a positive divergence implies nonlinearity or non-CP evolution or both. Similarly, the vorticity
\begin{eqnarray}
{\bm \omega} = \nabla \times (d {\bf r} /dt), \ \ 
\omega^a = \varepsilon^{abc}  \partial_b 
[ G^{cd}({\bf r}) r^d]
= \varepsilon^{abc} G^{cb}_{\! -} ({\bf r})
+ g \, \varepsilon^{abc} [\partial_b G^{cd}({\bf r})] \, r^d 
\label{qubit flow curl equation}
\end{eqnarray}
also has linear and nonlinear contributions ($\varepsilon$ is the Levi-Civita symbol and $g=1$). The $\varepsilon^{abc} G_{\! -}^{cb}({\bf r})$ term will contribute if $G({\bf r}) \in {\mathbb R}^{3 \times 3}$ has an antisymmetric ($|{\bf r}|$-conserving) part.

The divergence and vorticity faithfully characterize the velocity field, but don't adequately {\it quantify} the computational benefits of nonlinearity. This is because the velocity field describes how single states $X_\alpha$  follow their streamlines, but does not directly convey the {\it relative} motion between potential trajectories. For a more sensitive characterization we want to consider how pairs of states $(X_\alpha, X_\beta)$ transform under the channel. To further motivate this, consider a common setting for quantum algorithms, where a subroutine accepts as input a sequence of quantum states $(X_1 ,X_2, X_3, \cdots)$, then applies the same channel $\Lambda$ to each in order to learn something about those states or compute some function of those states. For example, we might  know that the states can only take values from a given set $\{ Y_1, Y_2, \cdots \}$, and we want to identify which. Previous authors \cite{MielnikJMP80,PhysRevLett.81.3992,BechmannPLA98,0502072,150706334,150305342} have noted the intriguing computational power afforded by the ability to {\it increase} the distinguishability between a pair of potential inputs $X_\alpha$ and  $X_\beta$, i.e., to increase their trace distance $ \| X_\alpha - X_\beta \|_1$, which is prohibited in linear CPTP channels. Let's examine this for a qubit in the Pauli basis: The differential of $\| X \|_p := [{\rm tr} (|X|^p)]^\frac{1}{p}$ for any square matrix $X$ is $d \| X \|_p = \| X \|_p^{1-p} \ {\rm tr} (|X|^{p-1} d |X|).$ Now let $X = X_\alpha - X_\beta = \frac{1}{2}( {\bf r}_{\alpha} - {\bf r}_{\beta}) \cdot {\bm \sigma} $ be the difference between a pair of qubit states with Bloch vectors  ${\bf r}_{\alpha,\beta} \in \CBALL$ and separation
$  \| X_\alpha - X_\beta \|_p = 2^{\frac{1}{p}-1}  \, | {\bf r}_{\alpha} -{\bf r}_{\beta} | $
measured in Schatten norm. Then
\begin{eqnarray}
\frac{d}{dt} \| X_\alpha - X_\beta \|_p = 2^{\frac{1}{p}-1} \,
\frac{{\bf r}_{\alpha} - {\bf r}_{\beta }}{ |{\bf r}_{\alpha} - {\bf r}_{\beta}|}
\cdot
\bigg(\frac{d{\bf r}_{\alpha}}{dt} - \frac{d{\bf r}_{\beta }}{dt} \bigg) \le 2^{\frac{1}{p}-1} \bigg| \frac{d{\bf r}_{\alpha}}{dt} - \frac{d{\bf r}_{\beta }}{dt} \bigg| 
\end{eqnarray}
characterizes the expansivity of the channel: $\frac{d}{dt} \| X_\alpha - X_\beta \|_p < 0$ means that the channel is strictly contractive on the pair, 
$\frac{d}{dt} \| X_\alpha - X_\beta \|_p = 0$ means it's  distance preserving on the pair, and $\frac{d}{dt} \| X_\alpha - X_\beta \|_p > 0$ means it's expansive. Expansivity allows for the distance between two nearby states $(X_\alpha , X_\beta)$ to increase. In the notation of (\ref{general qubit evolution equation}) the rate of change of state separation is
\begin{eqnarray}
\frac{d}{dt} \| X_\alpha - X_\beta \|_p = 2^{\frac{1}{p}-1} \,
\frac{ r^{a}_{\alpha} - r^{a}_{\beta}  }{ |{\bf r}_{\alpha} - {\bf r}_{\beta}| }
\bigg[ G^{a b}({\bf r}_\alpha) \, r_\alpha^b
- G^{a b}({\bf r}_\beta) \, r_\beta^b \bigg] .
\label{general qubit expansivity} 
\end{eqnarray}
Expanding about the midpoint $ {\bf R} = ( {\bf r}_\alpha +  {\bf r}_\beta )/2$ gives, to second order in $|{\bf r}_{\alpha}  -   {\bf r}_{\beta}|$,
\begin{eqnarray}
\frac{d}{dt} \| X_\alpha - X_\beta \|_p 
\approx 2^{\frac{1}{p}-1} \,
 \frac{ [G_{\! +}^{a b}( {\bf R})  \! + \!  K_{\! +}^{a b}({\bf R})]  
 (r^{a}_{\alpha} \! - \! r^{a}_{\beta} ) (r^{b}_{\alpha}  \! - \!     r^{b}_{\beta} )  }
 {|{\bf r}_{\alpha}  -   {\bf r}_{\beta}|}, \ \ 
 K^{a b}({\bf R}) := g R^c \, \partial_b G^{a c}({\bf R}) . 
 \ \ \ \ \ \ 
\label{qubit expansivity gradient expansion} 
\end{eqnarray}
Here the coupling $g=1$ is added to indicate the presence of nonlinearity. We note the two distinct sources of expansivity in (\ref{qubit expansivity gradient expansion}): 
The antisymmetric part of $G({\bf R})$ doesn't contribute to the expansivity, but positive eigenvalues in the symmetric part do. This is an alternative expression of the same results we found above for $d | {\bf r} |^2/dt$, and for the ${\rm tr}[G_{\! +}({\bf r})]$ term in the divergence.  The second term contributes to expansivity if the symmetric part of the matrix $K \in {\mathbb R}^{3 \times 3}$ has positive eigenvalues.

In the remainder of this section we apply this geometric characterization to a state-dependent CPTP channel with torsion \cite{MielnikJMP80,PhysRevLett.81.3992,150706334},
\begin{eqnarray}
\frac{dr^a}{dt} = G^{a b} \! ({\bf r}) \, r^b \! , \ \ 
G^{a b}({\bf r}) = g z J_z^{a b} , \ \ 
J_z = 
\begin{pmatrix}
0 & -1 & 0 \\
1 & 0 & 0 \\
0 & 0 & 0 \\
\end{pmatrix} \!  , \ \ 
g \in {\mathbb R}.
\label{torsion model}
\end{eqnarray}
Here $J_z$ is an SO(3) generator. $G({\bf r}) \in {\mathbb R}^{3 \times 3} $ is antisymmetric and hence $|{\bf r}|$-preserving. $G({\bf r})$ generates $z$ rotations with a rate that increases linearly with Bloch coordinate $z$, changing direction for $z<0$, a type of twist. The divergence (\ref{qubit flow divergence equation}) vanishes everywhere and the flow is incompressible. The vorticity (\ref{qubit flow curl equation}) is ${\bm \omega} = (-x, -y, 2 z ) g $. The $z$ component $\omega^3$ describes rigid body rotation within each plane of constant $z$, with a $z$-dependent frequency, while $\omega^{1}$ and  $\omega^{2}$ reflect the associated shear.
 We can use (\ref{qubit expansivity gradient expansion}) to discover expansive trajectories: In the torsion model (\ref{torsion model}), the matrix $K^{a b}({\bf R})$ defined in (\ref{qubit expansivity gradient expansion}) is
 \begin{eqnarray}
K = \frac{g}{2}
\begin{pmatrix}
0 & 0 & - (y_\alpha + y_\beta) \\
0 & 0 & (x_\alpha + x_\beta) \\
0 & 0 & 0 \\
\end{pmatrix}, \ \ 
K_{+} = \frac{g}{4}
\begin{pmatrix}
0 & 0 & - (y_\alpha + y_\beta) \\
0 & 0 & (x_\alpha + x_\beta) \\
 - (y_\alpha + y_\beta) & (x_\alpha + x_\beta) & 0 \\
\end{pmatrix}.
\end{eqnarray}
$K_{+}$ has eigenvalues 0 and
$\pm (|g|/4) \sqrt{ (x_\alpha + x_\beta)^2 
+  (y_\alpha + y_\beta)^2   } $.
For a pair of nearby states $X_\alpha, X_\beta$, their difference ${\bf r}_{\alpha} - {\bf r}_{\beta}$ is a short vector located at midpoint position ${\bf R} = 
({\bf r}_\alpha + {\bf r}_\beta)/2$. Expansive trajectories occur when
$ K^{ab}({\bf R}) (r^{a}_{\alpha} - r^{a}_{\beta} ) (r^{b}_{\alpha} - r^{b}_{\beta} ) 
 = 
\partial_b G^{a c}({\bf R}) R^c  (r^{a}_{\alpha} - r^{a}_{\beta} ) (r^{b}_{\alpha} - r^{b}_{\beta} ) $
is positive. In the torsion model this condition simplifies to 
\begin{eqnarray}
g  
\big[ R^x (y_{\alpha} - y_{\beta} ) - R^y (x_{\alpha} - x_{\beta} ) \big] (z_{\alpha} - z_{\beta} )  > 0.
\end{eqnarray}
Let ${\bf r}_\alpha = (\frac{1}{2}, \frac{\eta_y}{2} , \frac{\eta_z}{2})$ and ${\bf r}_\beta = (\frac{1}{2}, - \frac{\eta_y}{2} , - \frac{\eta_z}{2})$ be a pair of states with midpoint position ${\bf R} = (\frac{1}{2}, 0, 0)$ along the positive $x$ axis. The states are separated by $\eta_y$ in the $y$ direction and $\eta_z$ in the $z$ direction. For  nonzero $\eta_y$ and $\eta_z$, the two states  
move in opposite directions and separate at a rate
$\frac{d}{dt} \| X_\alpha - X_\beta \|_p = 2^{\frac{1}{p} - 2} g ( \eta_y \eta_z  / \! \sqrt{\eta_y^2 + \eta_z^2})$.

\subsection{General normalized PTP channels}
\label{general normalized ptp section}

Next we discuss channels with nonlinear positive $\phi$ and ${\rm tr}[\phi(X)] \neq 1 \ {\rm for \  some} \ X \in \RHO$, the most general PTP channels considered here. This class combines the nonunitary features of the NINO channels with the nonlinearity of state-dependent CPTP channels. Suppose we want to add linear dissipation/amplification to the torsion model (\ref{torsion model}) by adding a linear part $G$ to the generator. What are the allowed values of $G$? To answer this question, we use generators from the NINO evolution equation (\ref{def NINO evolution k=1}), namely $dX/dt = [ L_{-} , X] + \{ L_{+} , X \} +  \zeta_2 B X B^\dagger$, where we have included one jump operator $B$ and one nonjump operator $L$.\footnote{Here we assume that $\zeta_1 \! = \!1$;  the normalization condition on the $z_\alpha$ is then satisfied with $z_1 = 1$.}   
In the Pauli basis the first term in $dX/dt$ leads to $dr^a/dt = G^{ab} r^b$, with $G^{ab} = {\rm tr} ( \sigma^a L_{-} \sigma^b - \sigma^b L_{-} \sigma^a)/2$, resulting in an antisymmetric contribution to $G$. To see its connection with unitary dynamics, expand $L_{-} = -L_{-}^\dagger $ in the Pauli basis as $L_{-} = i(\xi_0 I + \xi_a \sigma^a)$, where $\xi_0, \dots, \xi_3 \in {\mathbb R}$ are real coordinates for $L_{-}$.
In this basis $G^{ab} = 2 \varepsilon^{abc} \xi_c  = - i \,  {\rm tr}(L_{-} \sigma^c) \,  \varepsilon^{abc}$, a real but otherwise arbitrary linear combination of SO(3) generators. Any linear antisymmetric $G$ can be implemented by controlling these generators. Similarly, the $\{ L_{+} , X \}$ term leads to a real symmetric $G^{ab} = {\rm tr} ( \sigma^a L_{+} \sigma^b + \sigma^b L_{+} \sigma^a)/2$ plus an inhomogeneous part $C^a =  {\rm tr} ( \sigma^a L_{+})$. Expanding $L_{+} = L_{+}^\dagger $ in the Pauli basis as $L_{+} = \xi_0 I + \xi_a \sigma^a$, where $\xi_0, \dots, \xi_3 \in {\mathbb R}$ are again real, leads to a diagonal matrix
$G  = 2 \xi_0 I_3 = {\rm tr}(L_{+}) I_3 $.  And the $ \zeta_2 \, B X B^\dagger$ term leads to $G^{ab} =  \zeta_2 \, {\rm tr} ( \sigma^a B  \sigma^b B^\dagger)/2$ and $C^a =  \zeta_2 \, {\rm tr} ( \sigma^a B B^\dagger)/2$ in the Pauli basis. Expanding $B \in \CMATRIX{2}$ as $B = \xi_0 I + \xi_a \sigma^a$, with $\xi_0, \dots, \xi_3 \in {\mathbb C}$ {\it complex} coordinates for $B$, we have
\begin{eqnarray}
G^{ab} = \zeta_2 \,  \big( |\xi_0|^2 - |\xi_1|^2 -  |\xi_2|^2  -  |\xi_3|^2 \big)  \delta^{ab} 
+ 2 \zeta_2 \,  {\rm Im} (\xi^*_0 \xi_c ) \varepsilon^{abc} 
+ 2 \zeta_2 \,  {\rm Re} (\xi^*_a \xi_b ) .
\end{eqnarray}
The first term is diagonal. The second term is antisymmetric (both  $L_{-}$  and $B$ contribute to  unitary evolution if this term is nonzero). The third term is symmetric. Let $\xi_0 = 0$; then
\begin{eqnarray}
G  =  - \zeta_2 \, (|\xi_1|^2 +  |\xi_2|^2  +  |\xi_3|^2) I
+ 2 \zeta_2 \,  {\rm Re}  \! \left[ \! 
\begin{pmatrix}
\xi^*_1 \\
\xi^*_2 \\
\xi^*_3
\end{pmatrix}
\! \!  \otimes \! \! 
\begin{pmatrix}
\xi_1 & \xi_2 & \xi_3 
\end{pmatrix}
\!  \right],
\end{eqnarray}
where $I$ is the identity. Consider now a {\it pair} of jump operators with the same $\zeta_2$ and coordinates $\xi = (1,1,0)$ and $(0,0,1)$:
\begin{eqnarray}
G _{(1,1,0)} = \zeta_2 
\begin{pmatrix}
0 & 2 & 0 \\ 
2 & 0 & 0 \\ 
0 & 0 & -2 
\end{pmatrix}, \ \ 
G _{(0,0,1)} = \zeta_2 
\begin{pmatrix}
-1 & 0 & 0 \\ 
0 & -1 & 0 \\ 
0 & 0 & 1 
\end{pmatrix}.
\end{eqnarray}
Combining them gives
\begin{eqnarray}
G _{(1,1,0)} + G _{(0,0,1)}  = \zeta_2 (2 \lambda_1 - I),
\ \ \lambda_1 = \begin{pmatrix}
0 & 1 & 0 \\ 
1 & 0 & 0 \\ 
0 & 0 & 0 
\end{pmatrix} \! ,
\end{eqnarray}
where $\lambda_1$ is a Gell-Mann matrix. Similarly, $G _{(1,0,1)} + G _{(0,1,0)}  = \zeta_2 (2 \lambda_4 - I)$ and
$G _{(0,1,1)} + G _{(1,0,0)}  = \zeta_2 (2 \lambda_6 - I)$, where $\lambda_4$ and $\lambda_6$ are Gell-Mann matrices. By combining Hamiltonian control with jump operator engineering, a large set of linear generators $G$ can be implemented.

\clearpage

\section{Fault-tolerant nonlinear state discrimination}
\label{state discrimination section}

In the remainder of the paper we consider an extension of the qubit torsion channel (\ref{torsion model}) that includes linear dissipation and amplification, applied to the problem of state discrimination  \cite{Helstrom1976,Holevo1982,08101970,0010114,BaeJPA15}. Using the techniques of Sec~\ref{general normalized ptp section}, jump operators are chosen such that 
\begin{eqnarray}
\frac{dX}{dt} = \frac{\sigma^a }{2}  
\bigg( \frac{dr^a}{dt} \bigg), \ \ 
\frac{dr^a}{dt}  = {\rm tr} \bigg( \! \frac{dX}{dt} \sigma^a \!\bigg)  \!  = G^{a b}({\bf r}) \, r^b = (m \, \lambda_4 - \gamma I + g z J_z)^{a b} r^b \! ,
\label{dissipative torsion model}
\end{eqnarray}
where $I$ is the $ 3 \! \times \! 3$ identity,  
\begin{equation}
\lambda_4 = 
\begin{pmatrix}
0 & 0 & 1 \\
0 & 0 & 0 \\
1 & 0 & 0 \\
\end{pmatrix} \!  ,
\ \ {\rm and} \ \ 
J_z = 
\begin{pmatrix}
0 & -1 & 0 \\
1 & 0 & 0 \\
0 & 0 & 0 \\
\end{pmatrix} \! .
\end{equation}
Here $\lambda_4$ is an SU(3) generator, $J_z$ is an SO(3) generator, and we sum over repeated indices $a, b \in (1,2,3)$. The dimensionless model parameters $m$, $\gamma$, and $g$ are real variables of either sign. The fixed-point equations are
\begin{eqnarray}
\frac{dx}{dt} &=& m z - \gamma x - g y z = 0, \\
\frac{dy}{dt} &=& - \gamma y + g x z = 0, \\
\frac{dz}{dt} &=& m x - \gamma z = 0.
\end{eqnarray}
The origin is always a fixed point,
$ {\bf r}^{\rm fp}_{0} = (0,0,0), $ although not always stable. Assuming ${\bf r} \neq (0,0,0)$, $\gamma \neq 0$, and eliminating $z$, the fixed point equations are
\begin{eqnarray}
m^2 - \gamma^2  = g m y, 
\label{first fixed point equation} \\
\gamma^2 y = gmx^2. 
\label{second fixed point equation}
\end{eqnarray}
If $g=0$, any fixed points must be confined to the $y=0$ plane. There are no additional fixed points unless $\gamma = \pm m$, in which case there is a set of fixed points ${\bf r}^{{\rm fp},g=0}_{z= \pm x}$ on the line ${\bf r}_{z=\pm x} = \{ (x,0,\pm x) : x \in {\mathbb R} \}$, shown in Fig.~\ref{separatrix figure}.  The settings $\gamma = \pm m$ are singular lines in the parameter space of the model. Manipulating these singularities in the presence of nonlinearity is the key to engineering useful information processing. 

\begin{figure}
\includegraphics[width=17.0cm]{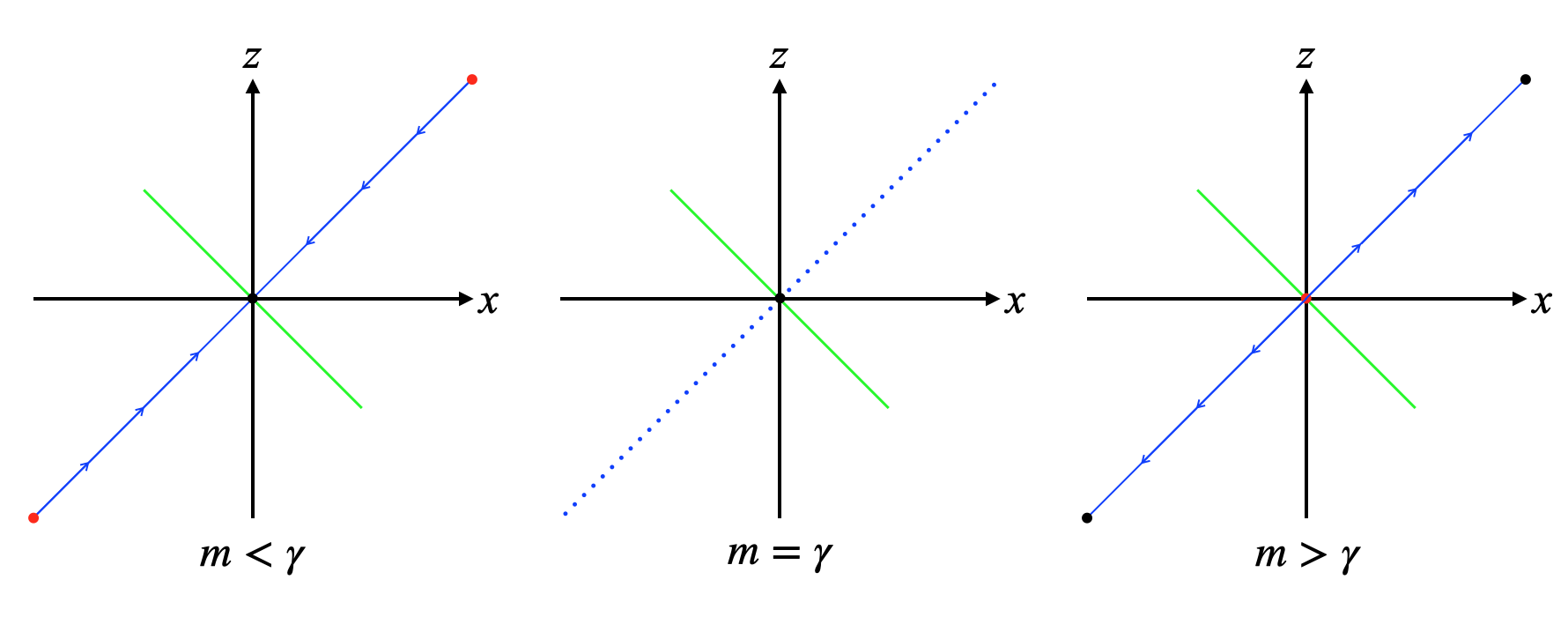} 
\caption{Illustration of the dynamics in the neighborhood of the origin for $m, \gamma \ge 0$. When $m < \gamma$, two unstable fixed points at infinity (red) feed the stable fixed point  ${\bf r}^{\rm fp}_{0}$ at the origin (black). However, if $m > \gamma$, ${\bf r}^{\rm fp}_{0}$  is unstable (red). The $\xi_{-}$ axis (green) is the separatrix. ${\bf r}^{\rm fp}_{0}$ feeds two new stable fixed points ${\bf r}^{\rm fp}_{+,-}$ (black). At the critical point $m = \gamma$, all points on the $\xi_{+}$ axis (blue dots) are fixed points. }
\label{separatrix figure}
\end{figure} 

When $g > 0$,  any fixed points must be confined to the plane $m y = (m^2 - \gamma^2)/g$. However (\ref{second fixed point equation}) requires $y$ to have the same sign as that of $m$. Therefore $m y >0 $, which is only possible when $m^2 > \gamma^2$.
Therefore, when $m^2  <  \gamma^2$, the only fixed point is ${\bf r}^{\rm fp}_{0}$, and this fixed point is stable for all $m^2  <  \gamma^2$. If instead the condition $m^2 > \gamma^2$ is satisfied, and $g > 0$, there is a pair of stable fixed points at
\begin{eqnarray}
{\bf r}^{\rm fp}_{\pm} = \bigg( \! \pm \frac{ |\gamma | }{g}
\sqrt{ \delta } , \ 
\frac{m}{g} \delta , \
 \pm \, {\rm sign}(\gamma) \frac{m}{g} 
\sqrt{ \delta } \bigg), \ \ 
\delta := \frac{m^2 - \gamma^2}{m^2} \in (0,\infty].
\label{stable fixed points}
\end{eqnarray}
For these fixed points to be contained within the Bloch ball requires $|g| > g_{\rm min}$, where $g_{\rm min} = \sqrt{ (\gamma^2 + m^2) \delta  + m^2 \delta^2 }$. The dynamics between fixed points ${\bf r}^{\rm fp}_{-}, {\bf r}_0^{\rm fp}$, and ${\bf r}^{\rm fp}_{+}$ can be understood as follows: When $g=0$ we have
\begin{eqnarray}
\frac{dx}{dt} &=& m z - \gamma x, \\
\frac{dy}{dt} &=& - \gamma y, \\
\frac{dz}{dt} &=& m x - \gamma z.
\end{eqnarray}
Note that the $y$ motion is decoupled from $x$ and $z$, and that it is always stable for $\gamma > 0$. Furthermore, the linearized model has an additional symmetry which becomes explicit after changing variables to
$ \xi_{\pm} = (z \pm x)/2$:
\begin{eqnarray}
\frac{d\xi_{+}}{dt} &=& (m -\gamma) \xi_{+}, 
\label{xi plus equation} \\
\frac{d\xi_{-}}{dt} &=& -  (m + \gamma) \xi_{-}.  
\label{xi minus equation}
\end{eqnarray}
The $\xi_{+}$ and $\xi_{-}$ variables are also decoupled. $\xi_{+}$ is the coordinate along the line $z=x$ mentioned above, and $\xi_{-}$ is the coordinate along the perpendicular line $z=-x$. Motion in the $\xi_{+}$ direction is stable for $m < \gamma$; in this case each point on the line $z=x$ flows to the fixed point ${\bf r}_0^{\rm fp}$ at the origin. However the $\xi_{+}$ motion becomes unstable when $m > \gamma$. In this regime ${\bf r}_0^{\rm fp}$ is unstable, and each point on the line $z=x$ (other than $z \! = \! x \! = \! 0$) flows outward to infinity. We can interpret this unstable case as having two stable fixed points at  $( \infty, 0, \infty)$ and $( -\infty, 0, -\infty)$, at the ends of the line $z=x$. By contrast, close to the singularity at $m = \gamma$, the perpendicular $\xi_{-}$ motion is stable unless $m$ and $\gamma$ are both negative. In this picture, the most important effect of the nonlinearity is to move the two stable fixed points at infinity to the finite positions (\ref{stable fixed points}). This is illustrated in Fig.~\ref{separatrix figure}. In Fig.~\ref{phases figure 09} we plot the trajectories for a cloud of randomly chosen initial states (red dots) within the Bloch ball $\CBALL$, close to the bifurcation but in the $m^2 < \gamma^2$ phase. In Fig.~\ref{phases figure 11} we show the same plot in the $m^2 > \gamma^2$ phase. These simulations further support the picture described above.

\begin{figure}
\includegraphics[width=8.0cm]{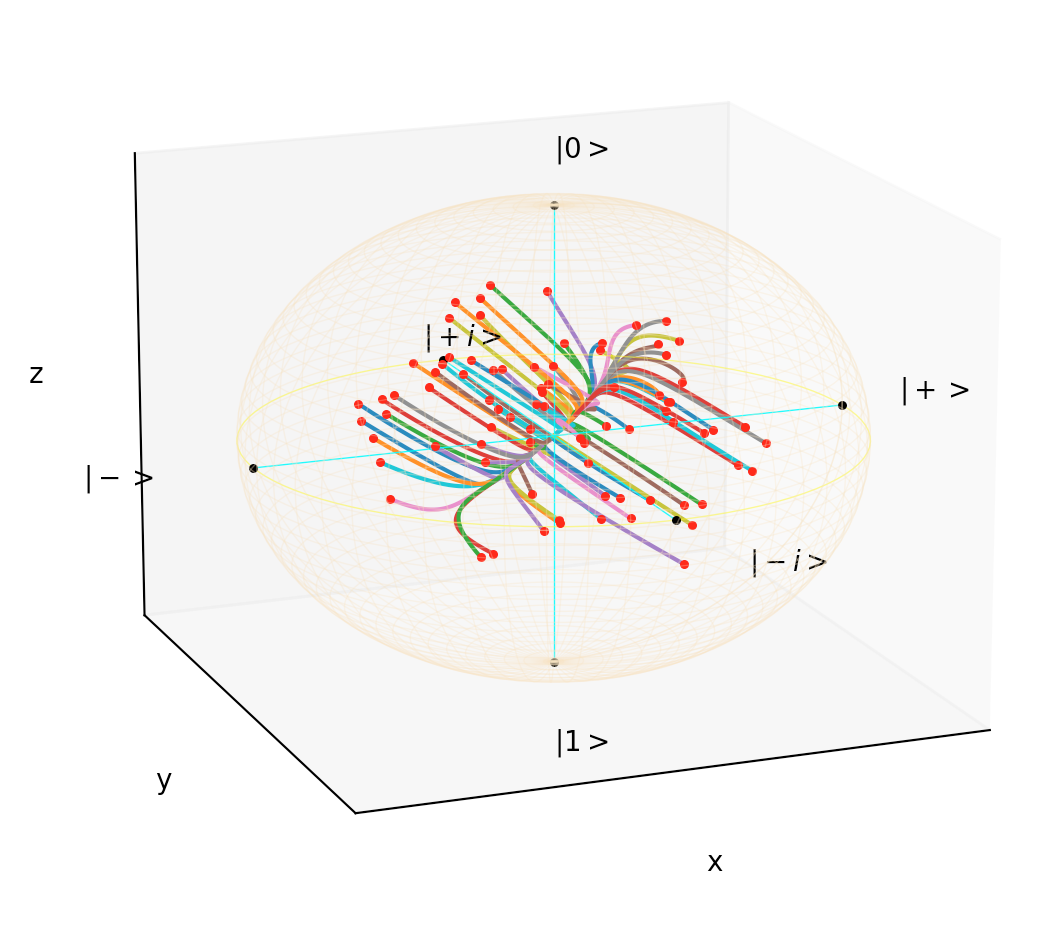} 
\caption{Simulation solutions of the torsion channel (\ref{dissipative torsion model}) with $\gamma = 1$, $m=0.9$, and $g=1$,
showing attraction to ${\bf r}^{\rm fp}_{0}$. The
$x$, $y$, and $z$ axes are Bloch vector coordinates.
Red dots indicate random initial conditions. 
The Bloch sphere is outlined in yellow.}
\label{phases figure 09}
\end{figure} 

\begin{figure}
\includegraphics[width=8.0cm]{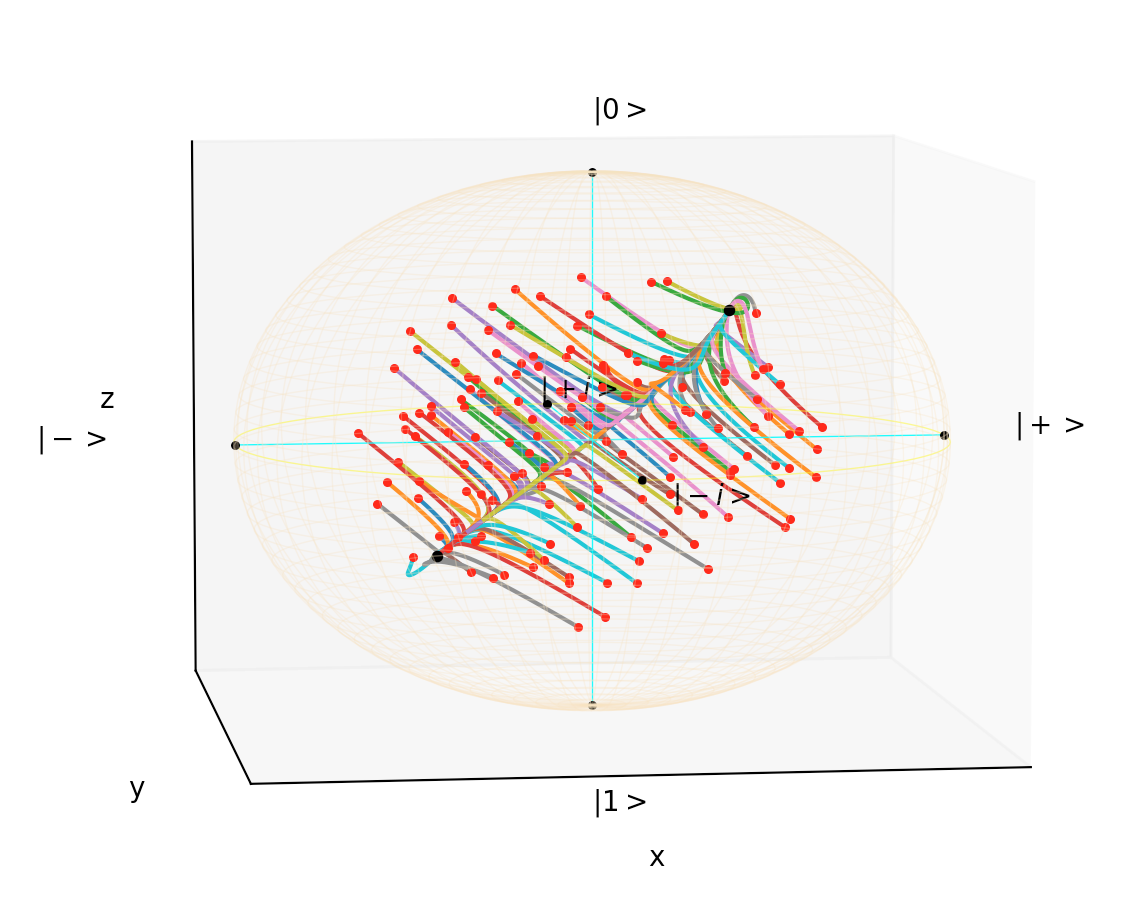} 
\caption{Simulation solutions to (\ref{dissipative torsion model}) with $\gamma = 1$, $m=1.1$, $g=1$, and $\delta = 0.2$, showing attraction to ${\bf r}^{\rm fp}_{\pm}$.}
\label{phases figure 11}
\end{figure} 

The $\xi_{+}$ dynamics near the unstable fixed point ${\bf r}_0^{\rm fp}$ can be used to achieve robust state discrimination with the exponential speedup supported by expansive nonlinear channels \cite{MielnikJMP80,PhysRevLett.81.3992,BechmannPLA98,0502072,150706334}. Points very close to ${\bf r}^{\rm fp}_{+}$ have $|{\bf r}| \ll 1$, so the nonlinearity can be neglected there. Equations (\ref{xi plus equation}) and (\ref{xi minus equation}) then apply to the dynamics near ${\bf r}_0^{\rm fp}$ even when $g\neq 0$. Consider the plane passing through the origin and perpendicular to the $\xi_{+}$ axis. The velocity field smoothly changes sign across this plane; i.e., it is a separatrix between basins of attraction for ${\bf r}^{\rm fp}_{+}$ and  ${\bf r}^{\rm fp}_{-}$.  Suppose that a qubit is prepared in a state from the set $\{ X_\alpha, X_\beta  \}$, with $X_\alpha$ and $X_\beta$ close in trace distance $ \epsilon = \|  X_\alpha -  X_\beta \|_1$ but on opposite sides of the separatrix. Then we implement a gate by turning on the nonlinearity for a time $t = O(1/g)$, during which $ X_\alpha$ and $X_\beta$ flow to different fixed points. After this evolution, the nonlinearity is turned off and the qubit is measured. This nonlinear gate leads to an exponential speedup if the initial separation is exponentially small: $\epsilon = 2^{-k}$ \cite{PhysRevLett.81.3992,0502072,150706334}.

Positivity of the dissipative torsion channel requires that the Bloch vector ${\bf r}$ remain in the Bloch ball ${|\bf r}| \le 1$. However the evolution equation  (\ref{dissipative torsion model}) does not itself enforce this condition. Therefore the positivity condition must be implemented dynamically through control of the qubit Hamiltonian, or added depolarization, and trajectories leaving the Bloch ball are regarded as unphysical. We note that a breakdown of positivity is expected in some open systems with initial system-environment entanglement that result in non-CP channels \cite{150305342}.
 
\section{Conclusion}
\label{conclusion section} 

\begin{quote}
“Our mistake is not that we take our theories too seriously, but that we do not take them seriously enough.” \ Steven Weinberg \cite{Weinberg1993} 
\end{quote}

In this paper, we introduced a classification for nonlinear  channels, and explored the computational power of three classes of associated nonlinear evolution equations (for qubits). Geometric characterizations of the dynamics, including flow divergence and expansivity, are shown to indicate the presence of particular forms of nonlinearity and non-complete positivity. This approach to classifying nonlinear maps appears to be new. The type B NINO channels \cite{BrodyPRL12,14057165,KowalskAP19,200309170}, and especially the type C state-dependent CPTP channels \cite{\EXPANSIVE,13033537,13030371,13107301,150706334,\NQMTHEORY,\EFFECTIVE,GisinJPA81},  have been discussed previously. To the best of our knowledge, type D channels, a main focus of this paper, have not been discussed previously. In this work, we did not try to justify the channels physically or provide microscopic models for them. Instead, in the spirit of Weinberg's famous quote, we impose only the minimal requirements of positivity and trace preservation, and ask what types of nonlinearity might be realized in principle, and what computational advantages they would provide. We see that engineering both nonlinearity and dissipation allows one to implement rich dynamics similar to that of classical nonlinear systems. Our main result is the identification of a novel phase where the Bloch ball separates into two basins of attraction, which can be used to implement fast quantum state discrimination 
\cite{PhysRevLett.81.3992,0502072,150706334} with intrinsic fault-tolerance. In particular, the states do not have to be initialized with high accuracy, but only in the appropriate basin of attraction. 

Although appealing, this gate has limitations: (i) First, finite experimental resolution and control will limit the smallest values of initial state separation $\epsilon$ achievable in practice. If the inputs to the discriminator are the outputs of a preceding process, they will also come with errors. (ii) Second, the fixed points ${\bf r}^{\rm fp}_{+,-}$ are not perfectly distinguishable. Although there is considerable flexibility in choosing their location, in practice they need to be well within the Bloch ball to ensure positivity. If the ${\bf r}^{\rm fp}_{+,-}$ are too close to the surface $|{\bf r}|=1$, trajectories approaching them may lead to unphysical solutions leaving the Bloch ball. (iii) And third, there will likely be errors associated with the effective model itself.

\section*{Acknowledgements}
 
This work was partly supported by the NSF under grant no. DGE-2152159. It is a pleasure to thank Andrew Childs for correspondence. The author also acknowledges anonymous referees for suggestions that have improved the paper.

%\bibliography{/Users/mgeller/Dropbox/bibliographies/CM,/Users/mgeller/Dropbox/bibliographies/MATH,/Users/mgeller/Dropbox/bibliographies/QFT,/Users/mgeller/Dropbox/bibliographies/QI,/Users/mgeller/Dropbox/bibliographies/books}

\bibliography{arxiv.bbl}

\end{document}